\documentclass{article}
\usepackage{amssymb}
\usepackage{amsmath,amsthm,amsfonts,amscd,mathrsfs}
\usepackage{color}
 \usepackage[]{hyperref}
\usepackage{cite}
\usepackage{geometry}
\usepackage{changes}
\usepackage{array}
\usepackage{indentfirst}
\usepackage{booktabs}
\newcommand{\tabincell}[2]{\begin{tabular}{@{}#1@{}}#2\end{tabular}}
\evensidemargin=\oddsidemargin \textwidth=158mm \textheight=220mm
\renewcommand\arraystretch{0.69}

\newtheorem{definition}{Definition}[section]
\newtheorem{lemma}{Lemma}[section]
\newtheorem{theorem}{Theorem}[section]
\newtheorem{corollary}[theorem]{Corollary}
\newtheorem{example}{Example}[section]
\newtheorem{remark}{Remark}[section]

\begin{document}

\title{Near MDS and near quantum MDS codes via orthogonal arrays}
\author{Shanqi Pang*, Chaomeng Zhang, Mengqian Chen, Miaomiao Zhang\\
College of Mathematics and Information Science, \\
Henan Normal University, Xinxiang, 453007, China\\
* Correspondence: shanqipang@126.com}
\date{}
\maketitle
\noindent{\bf Abstract~---~}  Near MDS (NMDS) codes are closely related to interesting objects in finite geometry and have nice applications in combinatorics and cryptography. But there are many unsolved problems about construction of NMDS codes. In this paper, by using symmetrical orthogonal arrays (OAs), we construct a lot of NMDS, $m$-MDS and almost extremal NMDS codes. We establish a relation between asymmetrical OAs and quantum error correcting codes (QECCs) over mixed alphabets. Since quantum maximum distance separable (QMDS) codes over mixed alphabets with the dimension equal to one have not been found in all the literature so far, the definition of a near quantum maximum distance separable (NQMDS) code over mixed alphabets is proposed. By using asymmetrical OAs, we obtain many such codes.

\noindent{\bf Key words~---~} orthogonal array; NMDS; NQMDS code over mixed alphabets.

\section{Introduction}

In digital communication, due to various interferences, errors occur during the transmission of information, which requires that the information is encoded so that it has the ability to self-correct. MDS codes are a kind of error correcting code with good performance. However, since the parameters of an MDS code are limited by the size of the field, it is desirable to study codes nearly meeting the Singleton bound with more flexible parameters \cite{ZHengCLiXWang2022}. For a linear code $C=[n, k, d]_s$ define $S(C)=n-k+1-d$. If $S(C)=S(C^\perp)=m$, we call $C$ is $m$-MDS. Particularly, if $S(C)=1$, $C$ is \rm almost MDS (AMDS), and $S(C)=S(C^\perp)=1$, $C$ is \rm near MDS (NMDS) \cite{JSuiQYueXLiDHuang2022}. An AMDS code and a linear orthogonal array are equivalent \cite{MADeBoer1996}. Thus AMDS and NMDS codes are valuable and interesting as they have special geometric properties \cite{ZHengCLiXWang2022}. The first NMDS code was the $[11,6,5]$ ternary Golay code discovered in 1949 by Golay \cite{MJEGolay1949}, which has applications in group theory and combinatorics. Some recent progress on theory and applications of NMDS codes were made in \cite{QWangZHeng2021,HTong2013,LJin2019,CDing2020,CDing2018,PTan2021,YZhou2009,RDodunekova1997,XKai2015,ZHengCLiXWang2022}. In \cite{CDing2020}, Ding and Tang constructed infinite families of NMDS codes which hold $t$-designs, $t=2,3,4$. Ding also constructed $t$-designs from some geometry codes containing AMDS ones \cite{CDing2018}. In \cite{PTan2021}, several families of NMDS codes which are both distance-optimal and dimension-optimal locally recoverable codes were studied. In \cite{YZhou2009}, the authors used NMDS codes to construct secret sharing schemes which have good security properties. The error detection capability of AMDS and NMDS codes was studied in \cite{RDodunekova1997} and conditions for the codes to be good for error detection were established. In \cite{XKai2015}, the authors constructed MDS symbol-pair codes from AMDS codes. In \cite{ZHengCLiXWang2022}, based on cyclic subgroups of $F_{q^2}^\ast$, the authors constructed MDS, NMDS and AMDS codes. There are still a lot of NMDS codes remaining unknown.

In particular, NMDS codes with parameters $[2q+k,k+1,2q-1]$ over $GF(q)$ are said to be almost extremal. Almost extremal NMDS codes with $k>q$ are all known. But the existence and construction of $k\leq q$ are still open \cite{QWangZHeng2021}. 

In this paper, we explicitly construct almost extremal NMDS codes through OAs such as $[6,3,3]$ NMDS code over $GF(2)$, $[8,3,5]$ NMDS code over $GF(3)$.

As in the classical transmission of data, it is inevitable that errors occur in quantum information processing \cite{PWShor1995}. Since QECCs \cite{PWShor1995,PWShor1996,PWShor1997} could fight against various quantum noises, it has been attracting a great deal of attentions \cite{PWShor1995,RLaflamme1996}. Numerous classes of QMDS codes over a single alphabet have been constructed mainly from Galois field, Euclidean construction, Hermitian construction and OAs \cite{HChen2005,GCohen1999,KFeng2002,MGrasslMRötteler2004,DHu2008,LJin2010,LJin20102,GGL2011,ZLi2008,RLi2010,AMSteane1999,AMSteane19992,PKSarvepalli2005,   XKai2013,BChen2015,SLi2016,WFang2019,RYan2023,HXu2022,FYang2023}. However, it is expected that we frequently face a more complicated situation that quantum resources, in which quantum information is encoded, have different dimensions. In particular, we often use hybrid systems \cite{Ati2007,Ggen2008,Pbl2010} with different dimensions to store, transmit, and process the quantum information. Thus it is quite necessary to generalize QECCs over a single alphabet to mixed alphabets \cite{ZWang2013}. Construction of such codes has become one of the most important tasks in quantum coding theory \cite{ZWang2013, PWShor1996}. The QECCs $((n, K, d))_{s_1,s_2,\ldots,s_n}$ over mixed alphabets have been studied in \cite{ZWang2013,FShi2021}, and the quantum Singleton bound is also generalized. In \cite{RYan2023}, Yan et al. obtained some QECCs over mixed alphabets based on OAs. However, other than the above papers \cite{ZWang2013, PWShor1996,FShi2021,RYan2023}, there are little work on QECCs over mixed alphabets particularly on QMDS codes because of their harder construction. Besides, in all the literature, QMDS codes over mixed alphabets are all for $K>1$ \cite{RYan2023,ZWang2013} while such codes for $K=1$ have not been found so far.

An orthogonal array ${\rm OA}(N,n,s_1^{n_1}s_2^{n_2}\ldots s_v^{n_v},k)$ of strength $k$ is an $N\times n$ matrix, having $n_i$ columns with $s_i$ levels, $i=1,2,\ldots,v$, $v$ is an integer, $n=\sum\limits_{i=1}^vn_i$, and $s_i\neq s_j$ for $i\neq j$, with the property that, in any $N\times k$ submatrix, all possible combinations of $k$ symbols appear equally often as a row. The orthogonal array is called a mixed orthogonal array if $v\geq2$. Otherwise, the array is called symmetrical. OAs play a prominent role in the design of experiments which were introduced by Rao \cite{CRRao1947, RYan2023}. As is often the case, they can be useful for quantum information theory. In recent years, many new classes of OAs, especially high strength OAs have been obtained \cite{SPang2021,LChen2017,YWang2016,YZhang2004,XZhang2020,YZhu2005}. The relationship among OAs, classical error correcting codes (CECCs), quantum uniform states and QECCs was further revealed \cite{DGoyeneche2016,ZZheng2021,HXu2022,DGoyeneche2014,SPang2019,XZhang2021,PDelsarte1973,RYan2023}. An ${\rm OA}(N,n,s_1^{n_1},\ldots,s_v^{n_v},k)$ having $n=n_1+\cdots+n_v$ columns is called an irredundant orthogonal array (IrOA), if every subset of $n-k$ columns contains a different sequence of $n-k$ symbols in every row \cite{DGoyeneche2014}. IrOAs play an important role in the construction of quantum uniform states and QECCs. A lot of symmetrical or asymmetrical IrOAs, quantum uniform states and QECCs over a single alphabet or mixed alphabets including QMDS codes have also been constructed \cite{XKaiSZhuPLi2014,RYan2023,HXu2022,FYang2023}. It is these new developments in OAs that suggest the possibility of constructing NMDS, almost extremal NMDS and NQMDS codes.

In this paper, we present sufficient and necessary conditions for a symmetrical OA to be an NMDS or $m$-MDS code. Then we construct a lot of NMDS codes including almost extremal NMDS codes and $m$-MDS codes. Further, we establish a relation between asymmetrical OAs and QECCs over mixed alphabets. In addition, a near quantum MDS (NQMDS) code is defined. From an ${\rm OA}(s^k,2k+1,s^{2k}2^1,k)$ for even $s$, we can construct an {\rm NQMDS} code $((2k+1,1,k+1))_{s^{2k}2^1}$ such as $((3,1,2))_{8^22^1}$, $((3,1,2))_{12^22^1}$, $((5,1,3))_{16^42^1}$, $((5,1,3))_{20^42^1}$.

The rest of this paper is organized as follows. In Section \ref{Preli}, we introduce some basic notations and useful results on OAs, CECCs, QECCs and NQMDS codes. Main results are given in Section \ref{Main}. In Section \ref{mMDS}, we present sufficient and necessary conditions for a symmetrical OA to be an NMDS or $m$-MDS code. And then we construct a lot of NMDS codes including almost extremal NMDS codes and $m$-MDS codes. In Section \ref{3.3}, we construct NQMDS codes over mixed alphabets through asymmetrical OAs. The paper is concluded in Section \ref{Conclu}.

\section{Preliminaries}\label{Preli}

First, the notations used in this paper are listed as follows.

Let $Z_s^n$ denote the $n$-dimensional space over a ring $Z_s=\{0, 1, \ldots, s-1\}$. When $s$ is a prime power, let $F_s$ be a Galois field containing $s$ elements with binary operations (+ and $\cdot$). If $A=(a_{ij})_{n\times m}$ and $B=(b_{uv})_{s\times t}$ with elements from a Galois field, the Kronecker sum $A\oplus B$ is defined as $A\oplus B=(a_{ij}+B)_{ns\times mt}$ where $a_{ij}+B$ represents the $s\times t$ matrix with entries $a_{ij}+b_{uv} (1\leq u\leq s, 1\leq v\leq t)$ and the Kronecker product $A\otimes B$ is defined as $A\otimes B=(a_{ij}\cdot B)_{ns\times mt}$ where $a_{ij}\cdot B$ represents the $s\times t$ matrix with entries $a_{ij}\cdot b_{uv}$ $(1\leq u\leq s,\ 1\leq v\leq t)$. Let $(\mathbb C^s)^{\otimes n}=\underbrace{{\mathbb C}^s\otimes  {\mathbb C}^s\otimes \cdots \otimes {\mathbb C}^s}\limits_{n}$.

Some basic knowledge about ${\rm OA}$, CECC and QECC is given.

\begin{definition}\cite{SPang2017}\label{minimal distance} Let $R_1,\ldots,R_N$ be the rows of an $N\times t$ matrix $A$, with entries at the $i$th column from $Z_{s_i}=\{0,1,\ldots,s_i-1\}$, where $s_i\geq 2$ and $i=1,2,\ldots,t$. The
Hamming distance $Hd(R_u, R_v)$ between $R_u = (a_{u1},\ldots, a_{ut})$ and $R_v = (a_{v1}, \ldots, a_{vt})$ is defined as follows:
$$Hd(R_u, R_v)=|\{r : 1\leq r\leq t, a_{ur}\neq a_{vr}\}|.$$
\end{definition}

In this paper, $md(L)$ denotes the minimum Hamming distance between two distinct rows of
an ${\rm OA}\ L$.

\begin{definition}\cite{SPang2021}\label{XLin2018} Let $A$ be the orthogonal array ${\rm OA}(N,n,s_1^{n_1}s_2^{n_2}\cdots s_v^{n_v},k)$ and $\{A_1,A_2,\ldots ,A_u\}$ be a set of orthogonal arrays ${\rm OA}(\frac{N}{u},n,s_1^{n_1}s_2^{n_2}\cdots s_v^{n_v},k_1)$. If
$\bigcup\limits_{i=1}^u A_i=A$ and $A_i\bigcap A_j=\varnothing$ for $i\neq j$, then $\{A_1,A_2,\ldots ,A_u\}$ is said to be an orthogonal partition of strength $k_1$ of $A$. In particular, when $k_1=0$, $\{A_1,A_2,\ldots ,A_u\}$ is still an orthogonal partition of $A$ of strength $0$.
\end{definition}

\begin{definition}\cite{FShi2021} \label{QQQ} An $((n, K, d))_s$ $\rm QECC$ has the quantum Singleton bound:
\begin{equation}\label{quantum bound 3}
  K\leq s^{n-2d+2}.
\end{equation}
An $((n,K,d))_{s_1,s_2,...,s_n}$ $\rm QECC$ satisfies the quantum Singleton bound:
\begin{equation}\label{quantum bound 1}
K\leq \min\{\prod\limits_{j\epsilon C}s_j\ |\ C\subset\{1,2,\ldots,n\}, |C|=n-2(d-1)\}\ 
\end{equation}
for $n\geq 2(d-1)+1$, and 
\begin{equation}\label{quantum bound 2}
K\leq 1 
\end{equation}
for $n=2(d-1).$
\end{definition}

A $\rm QECC$ that achieves the equality in Eq. (\ref{quantum bound 3}), Eq. (\ref{quantum bound 1}) or Eq. (\ref{quantum bound 2}) is called a quantum MDS (QMDS) code.

\begin{definition} \label{nqMDS} An $((n,K,d))_{s_1,s_2,...,s_n}$ is called a \rm {near quantum MDS (NQMDS)} code if 
$$K=\min\{\prod\limits_{j\epsilon C}s_j\ |\ C\subset\{1,2,\ldots,n\}, |C|=n-2(d-1)\}-1$$
for $n\geq2(d-1)+1$.
\end{definition}

\subsection{Important properties of OAs} 

\begin{lemma}\cite{SPang2019}{\label{distance}} The minimal distance of an ${\rm OA}(s^k, n, s, k)$ is $n-k+1$ for $s\geq2$ and $k\geq1$.
\end{lemma}

\begin{lemma}\label{ICMq} For a prime power $s$, let $(a_1, a_2,$ $\ldots, a_m)=((s)\oplus0_{s^{m-1}}, 0_s\oplus(s)\oplus0_{s^{m-2}}, \ldots, 0_{s^{m-1}}\oplus(s))$. $b_n=c_{i_1}a_{i_1}+\cdots+c_{i_{u-1}}a_{i_{u-1}}+a_{i_u}$ $(1\leq n\leq \frac{s^m-1}{s-1}-m, c_{i_v}\ \epsilon\ F_s, 1\leq u\leq m, 1\leq v\leq u-1)$. Then
$$A=(a_1, a_2, \ldots, a_m, b_1, b_2, \ldots, b_{\frac{s^m-1}{s-1}-m})$$
is a saturated orthogonal array ${\rm OA}(s^m, \frac{s^m-1}{s-1}, s, 2)$.
\end{lemma}

\begin{proof} It follows from linear independence of any two columns of $A$.
\end{proof}

\begin{remark} The method of Lemma \ref{ICMq} is called independent columns method, abbreviated {\rm IC} method.
\end{remark}

\begin{lemma}\cite{RYan2023} \label{mul} Assume that $A$ is an ${\rm OA}(N_1,n,s_1,t)$ with $md(A)=h_1$, and that $B$ is an ${\rm OA}(N_2,n,s_2,$
$t)$ with $md(B)=h_2$. Let $h=\min \{h_1, h_2\}$. Then there exists an ${\rm OA}(N_1N_2,n,s_1s_2,t)$ with $md=h$.
\end{lemma}

\begin{lemma}\cite{ZZheng2021}(Expansive replacement method)\label{replacement} Suppose $A$ is an ${\rm OA}$ of strength $k$ with column $1$ having $d_1$ levels and that $B$ also is an ${\rm OA}$ of strength $k$ with $d_1$ rows. After making a one-to-one mapping between the levels of column $1$ in $A$ and the rows of $B$, if each level of column $1$ in $A$ is replaced by the corresponding row from $B$, we can obtain an ${\rm OA}$ of strength $k$.
\end{lemma}

\subsection{Important properties of CECCs and QECCs} 

\begin{lemma}\cite{ASHedayat1999}\label{t+1} If $C$ is an $(n, N, d)_s$ $\rm CECC$ over $F_s$ with dual distance $d^\bot$, then the codewords of $C$ form the rows of an ${\rm OA}(N, n, s, d^\bot-1)$ with entries from $F_s$. Conversely, the rows of a linear ${\rm OA}(N, n, s, k)$ over $F_s$ form an $(n, N, d)_s$ $\rm CECC$ over $F_s$ with dual distance $d^\bot\geq k+1$. If the orthogonal array has strength $k$ but not $k+1, d^\perp$ is precisely $k+1$.
\end{lemma}

\begin{lemma} \cite{RYan2023} \label{QECC} Assume that there exists an ${\rm OA}(N,n,s,k)$ with md=$h$ and an orthogonal partition $\{A_1,\ldots, A_K\}$ of strength $k_0$. Let $d=\min \{k_0, h-1\}$. Then, there exists an $((n, K, d+1))_s$ $\rm QECC$.
\end{lemma}

\begin{lemma}\cite{FShi2021}\label{fei21} Let $Q$ be a subspace of $\mathbb{C}^{s_1} \otimes \mathbb{C}^{s_2}\otimes \cdots \otimes\mathbb{C}^{s_n}$. If $Q$ is an $((n,K, k+1))_{s_1,s_2,...,s_n}$ $\rm QECC$, then for any $k$ parties, the reductions of all states in $Q$ to the $k$ parties are identical. The converse is true. Further if $Q$ is pure, then any state in $Q$ is a $k$-uniform state. The converse is also true. In particular, when $s_1=s_2=\cdots=s_n$, $Q$ is an $((n, K, k+1))_{s_1}$ $\rm QECC$.
\end{lemma}

The Lemma \ref{fei21} can be regarded as the definition of a QECC $((n,K, k+1))_{s_1,s_2,...,s_n}$, where $n$ is the number of qudits, $K$ is the dimension of the encoding state, $k+1$ is the minimum distance, and $s_1,s_2,\ldots,s_n$ are the alphabet size.

\begin{lemma}\cite{DGoyeneche2016}\label{K-} If 
$L=\left(\begin{array}{ccccccc}
a_{11}& a_{12} & \ldots & a_{1n}\\
a_{21}& a_{22} & \ldots & a_{2n}\\
\vdots & \vdots &   & \vdots\\
a_{N1}& a_{N2} & \ldots & a_{Nn}\\
\end{array}\right)$ is an ${\rm IrOA}(N,n,s_1^{n_1}s_2^{n_2}\cdots s_v^{n_v},k)$, then the superposition of $N$ product states, $|\phi_{s_1^{n_1}s_2^{n_2}\cdots s_v^{n_v}}\rangle =|a_{11}a_{12}\ldots a_{1n}\rangle+|a_{21}a_{22}\ldots a_{2n}\rangle+\cdots+|a_{N1} a_{N2}\ldots a_{Nn}\rangle$ is a $k$-uniform state.
\end{lemma}

\section{Main Results} \label{Main}

In this section, we construct NMDS codes including almost extremal NMDS codes, $m$-MDS codes and NQMDS codes over mixed alphabets through OAs. Here we first give the relationship between OAs and QECCs. There exists a perfect match between the parameters of an ${\rm OA}(N,n,s_1^{n_1}s_2^{n_2}\cdots s_v^{n_v}, k)$, $A$, with an orthogonal partition $\{A_1, A_2,\ldots, A_K\}$ of strength $k_1$ and the parameters of an $((n, K, d))_{s_1^{n_1}s_2^{n_2}\cdots s_v^{n_v}}$ QECC, which is listed in Table \ref{table4}.

\renewcommand{\tablename}{Table}
 \begin{table}[htb]
 \renewcommand\arraystretch{1.1}
  \caption{ Correspondence between parameters of OAs and QECCs.}
   \label{table4}
     \tabcolsep=0.7cm
$$\begin{tabular}{c l l}
\hline
 &${\rm OA}(N,n,s_1^{n_1}s_2^{n_2}\cdots s_v^{n_v}, k)$&$\rm QECC$ $((n, K, d))_{s_1^{n_1}s_2^{n_2}\cdots s_v^{n_v}}$\\ \hline
$n$&Number of factors&Length of code\\ \hline
$K$&Number of partitioned blocks& Dimension of code\\ \hline
$d$& $\min\{k_1+1, md(A)\}$&Minimum distance of code\\ \hline
$s_1,s_2,\ldots,s_v$&Number of levels&Alphabet size\\ \hline
\end{tabular}$$
\end{table}

\subsection{Construction of MDS, NMDS and $m$-MDS CECCs through orthogonal arrays}\label{mMDS}

\begin{theorem} \label{mMDS} For a prime power $s$, suppose $A$ is an ${\rm OA}(s^k, n, s, t)$ $(n\geq k)$ constructed by {\rm IC} method. The rows of $A$ form a $C=[n, k, d]_s$ ${\rm CECC}$. Then

$(1)$. $C$ is {\rm MDS} if and only if the strength of $A$ is $k$;

$(2)$. $C$ is {\rm NMDS} if and only if $C$ is {\rm AMDS} and the strength of $A$ is $k-1$;

$(3)$. $C$ is {\rm $m$-MDS} if and only if $S(C)=m$ and the strength of $A$ is $k-m\ (k> m)$.
\end{theorem}

\begin{proof} Suppose $C$ is an ${\rm OA}(s^k, n, s, t)$ constructed from linear combination of $k$ independent columns $((s)\oplus0_{s^{k-1}}, 0_s\oplus(s)\oplus0_{s^{k-2}}, \ldots, 0_{s^{k-1}}\oplus(s))$. Because the row rank of $C$ is equal to its column rank, $C$ is a linear code $[n, k, d]_s$.

(1). If $C$ is MDS, from \cite{ASHedayat1999}, $C^\perp$ is also MDS. So $d^\perp=k+1$. It follows from Lemma \ref{t+1} that $t=k$. Conversely, if $t=k$, from Lemma \ref{distance} we have $d=n-k+1$. Thus $C$ is MDS.

(2). If $C=[n, k, d]_s$ is NMDS, then both $C$ and $C^\perp=[n, n-k, d^\bot]_s$ are AMDS. Thus $d^\perp=n-(n-k)=k$. It follows from Lemma \ref{t+1} that $t=d^\bot-1=k-1$. So $C$ is AMDS and $t=k-1$. Conversely, if $C$ is AMDS and $t=k-1$, we have $C^\bot$ is an $[n, n-k, k]_s$ CECC. i.e. $C^\perp$ is AMDS. Thus $C$ is NMDS.

(3). If $C=[n, k, d]_s$ is $m$-MDS, that is $S(C)=S(C^\perp)=m$ where $C^\perp=[n, n-k, d^\bot]_s$. Since $S(C^\bot)=n-(n-k)+1-d^\bot$, we have $d^\bot=k+1-m$. It follows from Lemma \ref{t+1} that $t=d^\bot-1=k-m$. So $S(C)=m$ and $t=k-m$. Conversely, if $S(C)=m$ and $t=k-m$, then $C^\bot$ is an $[n, n-k, k-m+1]_s$ CECC. Obviously, $S(C^\bot)=n-(n-k)+1-(k-m+1)=m$. That is, $S(C)=S(C^\bot)=m$. Thus $C$ is $m$-MDS.
\end{proof}

\begin{example} Let $s=2$ and $k=3$ in Theorem \ref{mMDS}. Let $(a_1,a_2,a_3)=((2)\oplus0_4, 0_2\oplus(2)\oplus0_2,0_4\oplus(2))$.

$(i)$. Suppose $t=k=3$. According to Lemma \ref{ICMq}, $A=(a_1,a_2,a_3)$ and $B=(a_1,a_2,a_3,a_1+a_2+a_3)$ are ${\rm OA}(8,3,2,3)$ and ${\rm OA}(8,4,2,3)$, respectively. Then, we have $[3,3,1]_2$ and $[4,3,2]_2$ {\rm MDS} codes through Theorem \ref{mMDS} $(1)$.

$(ii)$. Suppose $t=k-1=2$. From Lemma \ref{ICMq}, $A_1=(a_1,a_2,a_3,a_2+a_3)$, $A_2=(a_1,a_1+a_2,a_1+a_3,a_2+a_3,a_1+a_2+a_3)$, $A_3=(a_1,a_2,a_3,a_1+a_2,a_1+a_3,a_2+a_3)$ and $A_4=(a_1,a_2,a_3,a_1+a_2,a_1+a_3,a_2+a_3,a_1+a_2+a_3)$ are ${\rm OA}(8,4,2,2)$, ${\rm OA}(8,5,2,2)$, ${\rm OA}(8,6,2,2)$ and ${\rm OA}(8,7,2,2)$, respectively. Then, we have $[4,3,1]_2$, $[5,3,2]_2$, $[6,3,3]_2$ and $[7,3,4]_2$ {\rm NMDS} codes according to Theorem \ref{mMDS} $(2)$. 
\end{example}

\begin{example} Let $s=2$ and $k=4$ in Theorem \ref{mMDS}. Let $(a_1,a_2,a_3,a_4)=((2)\oplus0_8, 0_2\oplus(2)\oplus0_4,0_4\oplus(2)\oplus0_2, 0_8\oplus(2))$.

$(i)$. Suppose $t=k=4$. According to Lemma \ref{ICMq}, $A=(a_1,a_2,a_3,a_4)$ and $B=(a_1,a_2,a_3,a_4, a_1+a_2+a_3+a_4)$ are ${\rm OA}(16, 4, 2, 4)$ and ${\rm OA}(16, 5, 2, 4)$, respectively. Then, we have $[4,4,1]_2$ and $[5,4,2]_2$ {\rm MDS} codes through Theorem \ref{mMDS} $(1)$.

$(ii)$. Suppose $t=k-1=3$. From Lemma \ref{ICMq}, $A_1=(a_1,a_2,a_3,a_4,a_1+a_2+a_3)$, $A_2=(a_1,a_2,a_3,a_4,a_1+a_2+a_3,a_2+a_3+a_4)$, $A_3=(a_1,a_2,a_3,a_4,a_1+a_2+a_3,a_1+a_2+a_4,a_2+a_3+a_4)$ and $A_4=(a_1,a_2,a_3,a_4,a_1+a_2+a_3,a_1+a_2+a_4,a_1+a_3+a_4,a_2+a_3+a_4)$ are ${\rm OA}(16, 5, 2, 3)$, ${\rm OA}(16, 6, 2, 3)$, ${\rm OA}(16, 7, 2, 3)$ and ${\rm OA}(16, 8, 2, 3)$, respectively. Then, we have $[5,4,1]_2$, $[6,4,2]_2$, $[7,4,3]_2$ and $[8,4,4]_2$ {\rm NMDS} codes according to Theorem \ref{mMDS} $(2)$. 

In particular, $[8, 4, 4]_2$ is {\rm NMDS} self-dual.

$(iii)$. Suppose $m=2$ and $t=k-2=2$. From Lemma \ref{ICMq}, $A_1=(a_1,a_2,a_3,a_4,a_1+a_2,a_1+a_3)$, $A_2=(a_1,a_2,a_3,a_4,a_1+a_2,a_1+a_3,a_2+a_3)$, $A_3=(a_1,a_2,a_3,a_4,a_1+a_2,a_1+a_3,a_2+a_4,a_3+a_4)$ and $A_4=(a_1,a_2,a_3,a_4,a_1+a_2,a_1+a_3,a_2+a_4,a_3+a_4,a_1+a_2+a_3+a_4)$ are ${\rm OA}(16, 6, 2, 2)$, ${\rm OA}(16, 7, 2, 2)$, ${\rm OA}(16, 8, 2, 2)$ and ${\rm OA}(16, 9, 2, 2)$, respectively. Then, from Theorem \ref{mMDS} $(3)$, we can get $[6,4,1]_2$, $[7,4,2]_2$, $[8,4,3]_2$ and $[9,4,4]_2$ {\rm $2$-MDS} codes.
\end{example}

\begin{example} Let $s=3$ and $k=3$ in Theorem \ref{mMDS}. Let $(a_1,a_2,a_3)=((3)\oplus0_9, 0_3\oplus(3)\oplus0_3, 0_9\oplus(3))$.

$(i)$. Suppose $t=k=3$. According to Lemma \ref{ICMq}, $A=(a_1,a_2,a_3)$ and $B=(a_1,a_2,a_3,a_1+a_2+a_3)$ are ${\rm OA}(27, 3, 3, 3)$ and ${\rm OA}(27, 4, 3, 3)$, respectively. Then, from Theorem \ref{mMDS} $(1)$, there are two {\rm MDS} codes $[3,3,1]_3$ and $[4,3,2]_3$.

$(ii)$. Suppose $t=k-1=2$. From Lemma \ref{ICMq}, $A_1=(a_1,a_2,a_3,a_1+a_2,a_1+a_2+a_3,2a_1+a_3,2a_2+a_3,2a_1+a_2+a_3,a_1+2a_2+a_3)$, 
$A_2=(a_2,a_3,a_1+a_2,a_1+a_3,a_2+a_3,a_1+a_2+a_3,2a_1+a_2,2a_1+a_3)$, 
$A_3=(a_2,a_1+a_2,a_1+a_3,a_2+a_3,a_1+a_2+a_3,2a_1+a_2,2a_1+a_3)$, 
$A_4=(a_1+a_2,a_1+a_3,a_2+a_3,2a_1+a_2+a_3,a_1+2a_2+a_3,2a_1+2a_2+a_3)$,
$A_5=(a_1,a_2,a_3,a_1+a_2,a_1+a_3)$ and $A_6=(a_1,a_2,a_3,a_1+a_2)$ are ${\rm OA}(27, 9, 3, 2)$, ${\rm OA}(27, 8, 3, 2)$, ${\rm OA}(27, 7, 3, 2)$, ${\rm OA}(27, 6, 3, 2)$, ${\rm OA}(27, 5, 3, 2)$ and ${\rm OA}(27, 4, 3, 2)$, respectively. Then, according to Theorem \ref{mMDS} $(2)$, there exist six {\rm NMDS} codes $[9,3,6]_3$, $[8,3,5]_3$, $[7,3,4]_3$, $[6,3,3]_3$, $[5,3,2]_3$ and $[4,3,1]_3$.

\end{example}

Here, both $[6,3,3]_2$ and $[8,3,5]_3$ are almost extremal NMDS according to the definition of almost extremal NMDS codes.

\subsection{Construction of near quantum MDS codes over mixed alphabets through orthogonal arrays} \label{3.3}

\begin{theorem}\label{ORQU} Assume that there exists an ${\rm OA}(N,n,s_1^{n_1}s_2^{n_2}\cdots s_v^{n_v}, k)$ with $md=h$ and an orthogonal partition $\{A_1,\ldots, A_K\}$ of strength $k_0$. Let $d=\min\{k_0,h-1\}$. Then, there exists an $((n, K, d+1))_{s_1^{n_1}s_2^{n_2}\cdots s_v^{n_v}}$ {\rm QECC}.
\end{theorem}

\begin{proof} By Definition \ref{XLin2018}, the ${\rm OA}(N,n,s_1^{n_1}s_2^{n_2}$
$\cdots s_v^{n_v}, k)$ and $A_i\ (i=1,\ldots,K)$ are an ${\rm IrOA}(N,n,s_1^{n_1}s_2^{n_2}$
$\cdots s_v^{n_v},d)$ and an ${\rm IrOA}(\frac{N}{K},n,s_1^{n_1}s_2^{n_2}\cdots s_v^{n_v},d)$, respectively. From the link between IrOAs and uniform states in \cite{DGoyeneche2016} and $\{A_1,\ldots, A_K\}$, we can obtain $K$ $d$-uniform states $\{|\phi_1\rangle,\cdots, |\phi_K\rangle\}$, which can be used as an orthogonal basis. By Lemma \ref{fei21}, the complex subspace spanned by the orthogonal basis is an $((n, K, d+1))_{s_1^{n_1}s_2^{n_2}\cdots s_v^{n_v}}$ QECC.
\end{proof}

In fact, if there exists an ${\rm OA}(N, 2k, s, k)$ with $md=k+1$, from Lemma \ref{QECC}, there exists a $((2k,1,k+1))_s$ QECC which is a quantum MDS code according to Definition \ref{nqMDS}.

Sometimes, it is difficult to construct QECCs over mixed alphabets which achieve quantum Singleton bound, we will obtain near quantum MDS codes according to Definition \ref{nqMDS}. We have the following results.

\begin{theorem} \label{Nq}
  If there exists an ${\rm OA}(s^k,2k+1,s,k)$ for even $s$, then there exists an {\rm NQMDS} code $((2k+1,1,k+1))_{s^{2k}2^1}.$
\end{theorem}

\begin{proof} From Lemma \ref{distance}, the minimal distance of ${\rm OA}(s^k,2k+1,s,k)$ is $k+2$. We can obtain an ${\rm OA}(s^k,2k+1,s^{2k}2^1,k)$ with $md=k+1$ after an $s$-level column of ${\rm OA}(s^k,2k+1,s,k)$ is replaced by a two-level column through expansive replacement method in Lemma \ref{replacement}. By Theorem \ref{ORQU}, we have a QECC $((2k+1,1,k+1))_{s^{2k}2^1}$ which is also an NQMDS code according to Definition \ref{nqMDS}.
\end{proof}

\begin{corollary} \label{Nq1} Suppose there exist two arrays ${\rm OA}(s_1^{k},2k+1,s_1,k)$ and ${\rm OA}(s_2^{k},2k+1,s_2,k)$. Let $s=s_1s_2$ be even. Then, there exists an {\rm NQMDS} code $((2k+1,1,k+1))_{s^{2k}2^1}.$
\end{corollary}

\begin{proof} From Lemma \ref{mul}, we have an ${\rm OA}(s^{k},2k+1,s,k)$. From Theorem \ref{Nq}, we have the corollary is true.
\end{proof}

\begin{corollary}\label{cor1} Suppose there exist $m\geq 3$ arrays ${\rm OA}(s_1^{k},2k+1,s_1,k)$, ${\rm OA}(s_2^{k},2k+1,s_2,k)$, \ldots, ${\rm OA}(s_m^{k},2k+1,$
$s_m,k)$. Let $s=s_1s_2\cdots s_m$ be even. Then, there exists an {\rm NQMDS} code $((2k+1,1,k+1))_{s^{2k}2^1}$.
\end{corollary}
\begin{proof} Repeatedly using Corollary \ref{Nq1}, we have the corollary is true.
\end{proof}

\begin{theorem}\label{cor2} An {\rm NQMDS} code $((5,1,3))_{s^42^1}$ exists for even $s\geq 4$ and $s\neq 6$. 
\end{theorem}

\begin{proof}
From \cite{ASHedayat1999}, we have the following conclusions: An ${\rm OA}(s^2,k,s,2)$ exists if and only if $k-2$ pairwise orthogonal Latin squares of order $s$ exist; There exist $s-1$ pairwise orthogonal Latin squares for prime power $s$; There exist more than 2 pairwise orthogonal Latin squares of order $s\geq 12$ which is not a prime power. When $s=10$, we have an ${\rm OA}(100,5,10^42^1,2)$ with $md=3$ in \cite{WarrenKufeld}. So ${\rm OA}(s^2,5,s^42^1,2)$ with $md=3$ for $s\geq4$ and $s\neq6$ can be obtained after an $s$-level column of ${\rm OA}(s^2,5,s,2)$ is replaced by a two-level column through expansive replacement method in Lemma \ref{replacement}. The proof is complete.
\end{proof}

\begin{example}
  Let $s=4$ and $k=2$. We have an ${\rm OA}(16,5,4,2)$. Then, we can obtain an {\rm NQMDS} code $((5,1,3))_{4^42^1}$ through Theorem \ref{Nq}.
\end{example}

Table \ref{(3)} list plenty of {\rm NQMDS} codes constructed by Theorem \ref{Nq}, \ref{cor2}, Corollary \ref{Nq1}, \ref{cor1}.

\renewcommand{\tablename}{Table}

 \begin{table}[htbp]
 \renewcommand\arraystretch{0.6}
  \caption{Near quantum MDS codes}
   \label{(3)}
   \tabcolsep=0.3cm
$$\begin{tabular}{c c|l}
\bottomrule   
\multicolumn{2}{c|}{Parameters} & \multicolumn{1}{c}{NQMDS code} \\ \hline
        $k$ & $s$ & $((2k+1,1,k+1))_{s^{2k}2^1}$ \\ \midrule        
        1 & 4 & $((3,1,2))_{4^22^1}$ \\
        1 & 6 & $((3,1,2))_{6^22^1}$ \\
        1 & 8 & $((3,1,2))_{8^22^1}$ \\
        1 & 10 & $((3,1,2))_{10^22^1}$ \\
        1 & 12 & $((3,1,2))_{12^22^1}$ \\
        1 & 14 & $((3,1,2))_{14^22^1}$ \\
        1 & 16 & $((3,1,2))_{16^22^1}$ \\\hline
        1 & $s=2t$ $(t\geq9)$ & $((3,1,2))_{s^22^1}$ \\ \hline
        2 & 4 & $((5,1,3))_{4^42^1}$ \\
        2 & 8 & $((5,1,3))_{8^42^1}$ \\
        2 & 10 & $((5,1,3))_{10^42^1}$ \\
        2 & 12 & $((5,1,3))_{12^42^1}$ \\
        2 & 14 & $((5,1,3))_{14^42^1}$ \\
        2 & 16 & $((5,1,3))_{16^42^1}$ \\
        2 & 18 & $((5,1,3))_{18^42^1}$ \\
        2 & 20 & $((5,1,3))_{20^42^1}$ \\\hline
        2 & $s=2t$ $(t\geq11)$ & $((5,1,3))_{s^42^1}$ \\ \hline
        3 & 8 & $((7,1,4))_{8^62^1}$ \\
        3 & 16 & $((7,1,4))_{16^62^1}$ \\
        3 & 32 & $((7,1,4))_{32^62^1}$ \\
        3 & 56 & $((7,1,4))_{56^62^1}$ \\
        3 & 64 & $((7,1,4))_{64^62^1}$ \\
        3 & 72 & $((7,1,4))_{72^62^1}$ \\
        3 & 88 & $((7,1,4))_{88^62^1}$ \\
        3 & 104 & $((7,1,4))_{104^62^1}$ \\
        3 & 112 & $((7,1,4))_{112^62^1}$ \\ \hline
        3 & \tabincell{c}{$s=8\times2^u\times p_1^{v_1}\times p_2^{v_2}\times\cdots\times p_m^{v_m}$ \\ ($p_i$ is a prime and $p_i^{v_i}\geq 7)$} & $((7,1,4))_{s^62^1}$ \\\hline 
        4 & 8 & $((9,1,5))_{8^82^1}$ \\
        4 & 16 & $((9,1,5))_{16^82^1}$ \\
        4 & 32 & $((9,1,5))_{32^82^1}$ \\
        4 & 64 & $((9,1,5))_{64^82^1}$ \\
        4 & 72 & $((9,1,5))_{72^82^1}$ \\
        4 & 88 & $((9,1,5))_{88^82^1}$ \\
        4 & 104 & $((9,1,5))_{104^82^1}$ \\
        4 & 128 & $((9,1,5))_{128^82^1}$ \\\hline 
        4 & \tabincell{c}{$s=8\times2^u\times p_1^{v_1}\times p_2^{v_2}\times\cdots\times p_m^{v_m}$ \\ ($p_i$ is a prime and $p_i^{v_i}\geq 9)$} & $((9,1,5))_{s^82^1}$ \\\hline 
        5 & 16 & $((11,1,6))_{16^{10}2^1}$ \\
        5 & 32 & $((11,1,6))_{32^{10}2^1}$ \\
        5 & 64 & $((11,1,6))_{64^{10}2^1}$ \\
        5 & 128 & $((11,1,6))_{128^{10}2^1}$ \\
        5 & 176 & $((11,1,6))_{176^{10}2^1}$ \\
        5 & 208 & $((11,1,6))_{208^{10}2^1}$ \\\hline 
        5 & \tabincell{c}{$s=16\times2^u\times p_1^{v_1}\times p_2^{v_2}\times\cdots\times p_m^{v_m}$ \\ ($p_i$ is a prime and $p_i^{v_i}\geq 11)$} & $((11,1,6))_{s^{10}2^1}$ \\ \hline
        
        $\cdots$ & $\cdots$ & $\cdots$ \\ \hline
      \end{tabular}$$
\end{table}

\section{Conclusion} \label{Conclu}

In this paper, by using OAs, we construct NMDS codes, $m$-MDS codes and NQMDS codes over two distinct alphabets. In the future, we will study construction of QMDS and NQMDS codes over more distinct alphabets from asymmetrical OAs.

\section*{Acknowledgments}

\noindent$\mathbf{Funding:}$ This research was funded by the National Natural Science Foundation of China Grant number 11971004.

\noindent$\mathbf{Conflicts\ of\ Interest:}$ The authors declare no conflict of interest.


\begin{thebibliography}{99}


\bibitem{PWShor1996} C. H. Bennett, D. P. DiVincenzo, J. A. Smolin, W. K. Wootters, Mixed-state entanglement and quantum error correction, Phys. Rev. A 54 (5) (1996) 3824-3851.
\bibitem{MADeBoer1996} Mario A. de Boer, Almost MDS codes, Designs, Codes Cryptogr, 9 (2) (1996) 143-155.  
\bibitem{BChen2015} B. Chen, S. Ling, G. Zhang, Application of constacyclic codes to quantum MDS codes, IEEE Trans. Inf. Theory 61 (3) (2015) 1474-1484.
\bibitem{HChen2005} H. Chen, S. Ling, C. Xing, Quantum codes from concatenated algebraic-geometric codes, IEEE Trans. Inf. Theory 51 (8) (2005) 2915-2920.
\bibitem{GCohen1999} G. D. Cohen, S. B. Encheva, S. Litsyn, On binary constructions of quantum codes, IEEE Trans. Inf. Theory 45 (7) (1999) 2495-2498.
\bibitem{PDelsarte1973} P. Delsarte, An algebraic approach to the association schemes of coding theory, Philips Res. Repts Suppl. (10) (1973), Available online: \url{https://users.wpi.edu/~martin/RESEARCH/philips.pdf}.
\bibitem{RDodunekova1997} R. Dodunekova, S. M. Dodunekov, T. Kløve, Almost-MDS and near-MDS codes for error detection, IEEE Trans. Inf. Theory 43 (1) 285-290.     
\bibitem{CDing2020} C. Ding, C. Tang, Infinite families of near MDS codes holding $t$-Designs, IEEE Trans. Inf. Theory 66 (9) (2020) 5419-5428.    
\bibitem{CDing2018} C. Ding, Designs from linear codes, Singapore: World Scientific (2018).    
\bibitem{KFeng2002} K. Feng, Quantum codes $[[6,2,3]]_p$ and $[[7,3,3]]_p$ $(p\geq 3)$ exist, IEEE Trans. Inf. Theory 48 (8) (2002) 2384-2391.    
\bibitem{WFang2019} W. Fang, F. Fu, Some new constructions of quantum MDS codes, IEEE Trans. Inf. Theory 65 (12) (2019) 7840-7847.    
\bibitem{DGoyeneche2014} D. Goyeneche, K. Zyczkowski, Genuinely multipartite entangled states and orthogonal arrays, Phys. Rev. A 90 (2) (2014) 022316-1-022316-18. 
\bibitem{DGoyeneche2016} D. Goyeneche, J. Bielawski, K. Zyczkowski, Multipartite entanglement in heterogeneous systems, Phys. Rev. A 94 (1) (2016) 012346-1-012346-10. 
\bibitem{MJEGolay1949} M. J. E. Golay, Notes on digital coding, Proc. IRE 37 657 (1949).
\bibitem{DHu2008} D. Hu, W. Tang, M. Zhao, Q. Chen, S. Yu, C. H. Oh, Graphical nonbinary quantum error-correcting codes, Phys. Rev. A 78 (1) (2008) 012306-1-012306-11.    
\bibitem{ZHengCLiXWang2022} Z. Heng, C. Li, X. Wang, Constructions of MDS, near MDS and almost MDS codes from cyclic subgroups of $F_{q^2}^*$, IEEE Trans. Inf. Theory 68 (12) (2022) 7817-7831.     
\bibitem{ASHedayat1999}  A. S. Hedayat, N. J. A. Sloane, J. Stufken, Orthogonal Arrays: Theory and Applications, Springer: New York NY USA, 1999.
\bibitem{LJin2010} L. Jin, S. Ling, J. Luo, C. Xing, Application of classical Hermitian self-orthogonal MDS codes to quantum MDS codes, IEEE Trans. Inf. Theory 56 (9) (2010) 4735-4740.
\bibitem{LJin2019} L. Jin, H. Kan, Self-dual near MDS codes from elliptic curves, IEEE Trans. Inf. Theory 65 (4) (2019) 2166-2170.
\bibitem{LJin20102} L. Jin, C. Xing, A construction of new quantum MDS codes, IEEE Trans. Inf. Theory 60 (5) (2014) 2921-2925.
\bibitem{XKai2013} X. Kai, S. Zhu, New quantum MDS codes from negacyclic codes, IEEE Trans. Inf. Theory 59 (2) (2013) 1193-1197.    
\bibitem{XKai2015} X. Kai, S. Zhu, P. Li, A construction of new MDS symbol-pair codes, IEEE Trans. Inf. Theory 61 (11) (2015) 5828-5834.    
\bibitem{XKaiSZhuPLi2014} X. Kai, S. Zhu, P. Li, Constacyclic codes and some new quantum MDS codes, IEEE Trans. Inf. Theory 60 (4) (2014) 2080-2086. 
\bibitem{PWShor1997} E. Knill, R. Laflamme, Theory of quantum error-correcting codes, Phys. Rev. A 55 (2) (1997) 900-911.    
\bibitem{WarrenKufeld} W. F. Kufeld, Orthogonal Arrays, Available online: \url{https://support.sas.com/techsup/technote/ts723b.pdf} (2023).    
\bibitem{ZLi2008} Z. Li, L. Xing, X. Wang, Quantum generalized Reed-Solomon codes: Unified framework for quantum maximum-distance-separable codes, Phys. Rev. A 77 (1) (2008) 012308-1-012308-4.
\bibitem{RLi2010} R. Li, Z. Xu, Construction of $[[n,n-4,3]]_q$ quantum codes for odd prime power $q$, Phys. Rev. A 82 (5) (2010) 052316-1-052316-4.
\bibitem{SLi2016} S. Li, M. Xiong, G. Ge, Pseudo-cyclic codes and the construction of quantum MDS codes, IEEE Trans. Inf. Theory 62 (4) (2016) 1703-1710.
\bibitem{RLaflamme1996} R. Laflamme, C. Miquel, J. P. Paz, W. H. Zurek, Perfect quantum error correction code, Phys. Rev. Lett. 77 (1) (1996) 198-201.

\bibitem{Ati2007} A. Levy Yeyati, F. S. Bergeret, A. Martín-Rodero, T. M. Klapwijk, Entangled Andreev pairs and collective excitations in nanoscale superconductors, Nature Phys. 3 (2007) 455-459.
\bibitem{Ggen2008} G. P. Lansbergen, R. Rahman, C. J. Wellard, I. Woo, J. Caro, N. Collaert, S. Biesemans, G. Klimeck, L. C. L. Hollenberg, S. Rogge, Gate-induced quantum-confinement transition of a single dopant atom in a silicon FinFET, Nature Phys. 4 (8) (2008), 656-661.
\bibitem{GGL2011} Giuliano G. La Guardia, New quantum MDS codes, IEEE Trans. Inf. Theory 57 (8) (2011) 5551-5554.
\bibitem{SPang2019} S. Pang, X. Zhang, X. Lin, Q. Zhang, Two and three-uniform states from irredundant orthogonal arrays, NPJ Quantum Inf. 5 (52) (2019) 1-10.
\bibitem{SPang2021} S. Pang, J. Wang, D. K. J. Lin, M. Liu, Construction of mixed orthogonal arrays with high strength, Ann. Statist. 49 (5) (2021) 2870-2884.
\bibitem{SPang2017} S. Pang, R. Yan, S. Li, Schematic saturated orthogonal arrays obtained by using the contractive replacement method, Commun. Statist. Theory Methods 46 (18) (2017) 8913-8924.
\bibitem{ZZheng2021} S. Pang, X. Zhang, S. Fei, Z. Zheng, Quantum $k$-uniform states for heterogeneous systems from irredundant mixed orthogonal arrays, Quantum Inf. Process. 20 (4) (2021) 156-1-156-46. 
\bibitem{LChen2017} S. Pang, L. Chen, Generalized Latin matrix and construction of orthogonal arrays, Acta Math. Appl. Sinica 33 4 (2017) 1083-1092.  
\bibitem{YWang2016} S. Pang, Y. Wang, G. Chen, J. Du, The existence of a class of mixed orthogonal arrays, IEICE Trans. Fundam. Electron. Commun. Comput. Sci. E99-A (4) (2016) 863-868.  
\bibitem{YZhang2004} S. Pang, Y. Zhang, S. Liu, Further results on the orthogonal arrays obtained by generalized Hadamard product, Statist. Probab. Lett. 68 (1) (2004) 17-25. 
\bibitem{XZhang2020} S. Pang, X. Zhang, Q. Zhang, The Hamming distances of saturated asymmetrical orthogonal arrays with strength $2$, Comm. Statist. Theory Methods 49 (16) (2020) 3895-3910.
\bibitem{YZhu2005} S. Pang, Y. Zhu, Y. Wang, A class of mixed orthogonal arrays obtained from projection matrix inequalities, J. Inequal. Appl. 241 (2015), 1-9.  
\bibitem{XZhang2021} S. Pang, X. Zhang, J. Du, T. Wang, Multipartite entanglement states of higher uniformity, J. Phys. A: Math. Theor. 54 (2021) 015305-1-015305-10.    
\bibitem{PangWang2018} S. Pang, X. Wang, J. Wang, J. Du, M. Feng, Construction and count of $1$-resilient rotation symmetric Boolean functions, Inf. Sci. 450 (2018) 336-342. 
\bibitem{HXu2022} S. Pang, H. Xu, M. Chen, Construction of binary quantum error-correcting codes from orthogonal array, Entropy 24 (2022) 1000-1-1000-12.
\bibitem{FYang2023} S. Pang, F. Yang, R. Yan, J. Du, T. Wang, Construction of quaternary quantum error-correcting codes via orthogonal arrays, Front. Phys. 11 (2023) 1148398-1-1148398-6.
\bibitem{Pbl2010} P. Rabl, S. J. Kolkowitz, F. H. L. Koppens, J. G. E. Harris, P. Zoller, M. D. Lukin, A quantum spin transducer based on nanoelectromechanical resonator arrays, Nature Phys. 6 (2010) 602-608. 
\bibitem{MGrasslMRötteler2004} M. Rotteler, M. Grassl, T. Beth, On quantum MDS codes, Proc. Int. Symp. Inf. Theory (2004) 356-356.    
\bibitem{CRRao1947} C. R. Rao, Factorial experiments derivable from combinatorial arrangements of arrays, Suppl. J. R. Statist. Soc. 9 (1) (1947) 128-139. 
\bibitem{FShi2021} F. Shi, M. Li, L.Chen, X. Zhang, $k$-uniform quantum information masking, Phys. Rev. A 104 (3) (2021) 032601-1-032601-8.  
\bibitem{JSuiQYueXLiDHuang2022} J. Sui, Q. Yue, X. Li, D. Huang, MDS, near-MDS or 2-MDS self-dual codes via twisted generalized Reed-Solomon codes, IEEE Trans. Inf. Theory 68 (12) (2022) 7832-7841.        
\bibitem{PWShor1995} P. W. Shor, Scheme for reducing decoherence in quantum computer memory, Phys. Rev. A 52 (4) (1995) 2493-2496.
\bibitem{AMSteane1999} A. M. Steane, Quantum reed-muller codes, IEEE Trans. Inf. Theory 45 (5) (1999) 1701-1703.
\bibitem{AMSteane19992} A. M. Steane, Enlargement of calderbank-shor-steane quantum codes, IEEE Trans. Inf. Theory 45 (7) (1999) 2492-2495.
\bibitem{PKSarvepalli2005} P. K. Sarvepalli, A. Klappenecker, Nonbinary quantum reed-muller codes, Proc. Int. Symp. Inf. Theory (2005) 1023-1027.
\bibitem{PTan2021} P. Tan, C. Fan, C. Ding, Z. Zhou, The minimum linear locality of linear codes, ArXiv abs/2102.00597 (2021). 
\bibitem{HTong2013} H. Tong, D. Ying, Quasi-cyclic NMDS codes,  Finite Fields their Appl.  24 (2013) 45-54.
\bibitem{ZWang2013} Z. Wang, S. Yu, H. Fan, C. H. Oh, Quantum error-correcting codes over mixed alphabets, Phys. Rev. A 88 (2) (2013) 022328-1-022328-7.     
\bibitem{QWangZHeng2021} Q. Wang, Z. Heng, Near MDS codes from oval polynomials, Discrete Math. 344 (4) (2021) 112277-1-112277-10.     
\bibitem{RYan2023} R. Yan, S. Pang, M. Chen, F. Yang, Quantum error-correcting codes based on orthogonal arrays, Entropy 25 (4) (2023) 680-1-680-18. 
\bibitem{YZhou2009} Y. Zhou, F. Wang, Y. Xin, S. Luo, S. Qing, Y. Yang, A secret sharing scheme based on near-MDS codes, IEEE Int. Conf. Netw. Infrastruct. Digit. Content (2009) 833-836.

\end{thebibliography}
\end{document}